\documentclass[a4paper,UKenglish,numberwithinsect]{lipics}
 
\usepackage{microtype}
\usepackage[utf8]{inputenc}
\usepackage{amssymb,amsmath, amsthm}
\usepackage{xspace}
\usepackage{tikz}
\usetikzlibrary{arrows,positioning,backgrounds,fit,shapes,calc,scopes}

\usepackage{graphicx}

\title{Partitioning graphs into induced subgraphs\footnote{Research supported by the CE-ITI grant project P202/12/G061 of GA ČR, by GAUK project 1784214 and by the project SVV--2015--260223.}}
\titlerunning{Partition into induced subgraphs} 

\author[1]{Dušan Knop}
\affil[1]{Department of Applied Mathematics, Charles University, Prague, Czech Republic\\
  \texttt{knop@kam.mff.cuni.cz}}
\authorrunning{D. Knop} 

\Copyright{Dušan Knop}

\subjclass{G.2.2 Graph Theory}
\keywords{induced path partition, modular-width, parameterized algorithm}

\DeclareMathOperator{\vc}{{\rm vc}}
\DeclareMathOperator{\mw}{{\rm mw}}
\DeclareMathOperator{\tc}{{\rm tc}}
\DeclareMathOperator{\nd}{{\rm nd}}
\DeclareMathOperator{\pw}{{\rm pw}}
\DeclareMathOperator{\tw}{{\rm tw}}
\DeclareMathOperator{\cw}{{\rm cw}}

\DeclareMathOperator{\poly}{{\rm poly}}
\DeclareMathOperator{\diam}{{\rm diam}}

\newcommand{\N}{\mathbb{N}}

\newcommand{\BigO}{\mathcal{O}}

\newcommand{\FPT}{{\sf FPT}\xspace}
\newcommand{\NP}{{\sf NP}\xspace}
\newcommand{\W}{{\sf W}\xspace}
\newcommand{\XP}{{\sf XP}\xspace}
\newcommand{\coNPpoly}{{\sf coNP/poly}\xspace}

\newcommand{\TrianglePartition}{\textsc{Triangle Partition}\xspace}
\newcommand{\SetPartition}{\textsc{Set Partition}\xspace}
\newcommand{\InducedHPartition}{\textsc{Partition into $H$}\xspace}
\newcommand{\HPartition}{\textsc{Partition into subgraphs of $H$}\xspace}
\newcommand{\InducedHPartitionH}[1]{\textsc{Partition into #1}\xspace}
\newcommand{\EquitableColoring}{\textsc{Equitable Coloring}\xspace}

\newcommand{\PerfectMatching}{\textsc{Perfect Matching}\xspace}

\theoremstyle{definition}\newtheorem{problem}[theorem]{Problem}
\theoremstyle{plain}\newtheorem{proposition}[theorem]{Proposition}
\theoremstyle{plain}\newtheorem{conjecture}[theorem]{Conjecture}

\newenvironment{prob}[1]{\begin{problem}[#1]~\\\begin{tabular}{rp{.8\textwidth}}}{\end{tabular}\end{problem}}

\begin{document}

\maketitle

\begin{abstract}
  We study the \InducedHPartition problem from the parameterized complexity point of view. In the \InducedHPartition problem the task is to partition vertices of a graph $G$ into sets $V_1,V_2,\dots,V_n$ such that the graph $H$ is isomorphic to the subgraph of $G$ induced by each set $V_i$ for $i = 1,2,\dots,n.$ The pattern graph $H$ is fixed.

For the parametrization we consider three distinct structural parameters of the graph $G$---namely the tree-width, the neighborhood diversity, and the modular-width. For the parametrization by the neighborhood diversity we obtain an \FPT algorithm for every graph $H.$
For the parametrization by the tree-width we obtain an \FPT algorithm for every connected graph $H.$ Finally, for the parametrization by the modular-width we derive an \FPT algorithm for every prime graph $H.$

\end{abstract}

\section{Introduction}
We begin with the definition of the \InducedHPartition problem. We will then present the problem in the light of some well-known problems from computation complexity---for example \PerfectMatching or \EquitableColoring---thus demonstrating it as a natural generalisation of these problems. Finally, we give the summary of our results presented in this paper.

\subsection{The \InducedHPartition problem}

For graphs $G = (V,E),H=(W,F)$ with $|V| = |W|\cdot n,$ we say that it is possible to {\em partition $G$ into copies of $H$} if there exist sets $V_1,V_2,\dots,V_n$ such that
\begin{itemize}
  \item $|V_i| = |W|$ for every $i = 1,2,\dots, n,$ 
  \item $\bigcup_{i = 1}^n V_i = V,$ and
  \item $G[V_i] \simeq H$ for every $i = 1,2,\dots, n,$
\end{itemize}
where by $G[V_i]$ we mean the subgraph of $G$ induced by the set of vertices $V_i.$

\begin{prob}{\InducedHPartition}
FIXED: & Template graph $H.$ \\
INPUT: & Graph $G$  with $|G| = n\cdot|H|$ for an integer $n.$  \\
QUESTION: & Is there a partition of $G$ into copies of $H?$ \\
\end{prob}

The complexity of the \InducedHPartition problem has been studied by Hell and Kirkpatrick~\cite{HellKirkpatrick78} and has been proven to be \NP-complete for any fixed graph $H$ with at least $3$ vertices---they have studied the problem under a different name as the {\sc Generalized Matching} problem. There are applications in the printed wiring board design~\cite{Hope72} and code optimization~\cite{BoeschGimpel}.

Some variants of this problem are studied extensively in graph theory. For example when $H = K_2$ the problem \InducedHPartitionH{$K_2$} is the well-known \PerfectMatching problem, which can be solved in polynomial time due to Edmonds~\cite{Edmonds65}---the algorithm works even for the optimization version, when one tries to maximize the number of copies of $K_2$ in $G.$ The characterisation theorem for $H = K_2,$ that is a characterisation of graphs admitting a prefect matching is known due to Tutte~\cite{Tutte47}.

Another frequently studied case of our problem is the \InducedHPartitionH{$K_3$} problem---also known as the \TrianglePartition problem. The \TrianglePartition problem arises as a special case of the \SetPartition problem (also known to be \NP-complete~\cite{garey-johnsonNP}). Gajarský et al.~\cite{GLO13} pointed out that the parametrized complexity of the \TrianglePartition problem parametrized by the tree-width of the input graph was not resolved so far.

The last, but not least, example of a well know problem which can be viewed as a special case of the \InducedHPartition problem is the \EquitableColoring problem. The task is to color the vertices of an $n$ vertex graph with exactly $k$ colors, such that vertices connected by an edge receive different colors and the resulting color classes have equal sizes. It is easy to see that the \EquitableColoring problem is the \InducedHPartition problem with the edgeless graph on $n/k$ vertices.

Very similar application can be found as the so called \emph{$\ell$-bounded vertex colorings}, where the task is to find a coloring of a graph $G$ with prescribed number of colors such each color is used at most $\ell$-times. This problem allows a straightforward reduction to the \EquitableColoring by inserting a suitable number of isolated vertices. The connection between these two problems was also studied from the parameterized complexity point of view~\cite{BodlaenderFomin:EqColTW}---an \XP algorithm is obtained for parametrization by the tree-width. Here a special case of this problem is again equivalent to the \InducedHPartition problem with $H\simeq k\cdot K_1.$ For this problem a polynomial time algorithm is known for trees~\cite{JarvisZhou01}. Upper and lower bounds on the number of colors are known for general graphs~\cite{HansenHK93}.


\subsubsection{Parameterized complexity results}
	When dealing with an \NP-hard problem it is usual to study the problem in the framework of parameterized complexity. While in the previous section we have introduced several problems of the classical complexity, here we give references to parameterized results for these and related problems. 

A similar but more general problem (called the {\sc MSOL Partitioning} problem) was studied by Rao~\cite{Rao07}. Here the task is to partition vertices of the graph $G$ into several sets $A_1,A_2,\dots,A_r$ such that $\varphi(A_i)$ holds for every $i = 1,2,\dots,r,$ where $\varphi(\cdot)$ is an MSO$_1$ formula with one free set variable. If the number $r$ and the clique-width $\cw$ are fixed then the algorithm runs in polynomial time and hence the problem belongs to an \XP class with parametrization by the clique-width.

It was shown by Fellows et al.~\cite{Fellows11:Colorful} that \EquitableColoring problem is \W[1]-hard parameterized by the tree-width of the input graph and the number of colors. This also proves (together with the fact that clique-width of a complement of the graph $G$ is the same as of the graph $G$) that the \InducedHPartition problem is \W[1]-hard even for $H = K_k$ and parametrization by clique-width of a graph.

\subsection{Our contribution}
Our first algorithm is based on the following theorem. Here we would like to point out, that even though the result follows easily, the application is not straightforward. A usual first step in the design of an \FPT algorithm for the tree-width is to use the celebrated theorem of Courcelle~\cite{Courcelle90}. Fortunately, it is possible to use the theorem for the \InducedHPartition problem when $H$ is a connected graph.

\begin{theorem}\label{thm:PartitionMSOFormula}
For any fixed connected graph $H$ the \InducedHPartition problem is expressible as an $MSO_2$ formula.
\end{theorem}

The proof of this theorem is rather technical and is contained in Section~\ref{sec:connGIPisMSO2}.

As the first algorithm is for graphs with bounded tree-width and though a sparse class of graphs, we also analyze some variants of the \InducedHPartition problem for a particular class of dense graphs. 
Many parameters are suitable for dense graph classes, such as neighborhood diversity and modular-width (we give formal definitions in Section~\ref{sec:StructuralGraphParameters}).

\begin{theorem}\label{thm:PartitionModularWidth}
For any fixed graph $H$ with $\nd(H) = |H|$ and a graph $G$ the \InducedHPartition problem belongs to class the \FPT class when parametrized by $\mw(G).$
\end{theorem}

We derive the result using integer linear programming in a fixed dimension, which can be solved by a parameterized routine~\cite{lenstra83, FT87}. It is worth to mention that even though the condition $\nd(H) = |H|$ may seem very restrictive this class of graphs contains for example paths $P_k$ for $k\ge 4$ and cycles $C_k$ for $k\ge 5.$ Applications of the \InducedHPartition problem with $H$ being a path may be found in code optimization~\cite{BoeschGimpel}.

We say that a graph $H$ is a {\em prime graph} if $H$ fulfills the condition $\nd(H) = |H|.$ The class of primal graphs is thoroughly studied in the context of modular decompositions.

We prove that the \InducedHPartition problem does not have polynomial kernel parameterized by modular-width for any reasonable graph $H.$ More precisely, we prove the following.

\begin{theorem}\label{thm:IHPKernelRefuteMW}
Let $H$ be a graph for which (unparameterized version of) the \InducedHPartition problem is \NP-hard. There is no polynomial kernel routine for the \InducedHPartition problem with parametrization $\mw(G)$ for input graph $G$ unless \NP$\subseteq$ \coNPpoly.
\end{theorem}

By using similar techniques we can prove that for every fixed $H$ the \InducedHPartition problem can be solved efficiently on graphs with bounded neighborhood diversity.
This is an exhaustive extension of Theorem~\ref{thm:PartitionModularWidth}.

\begin{theorem}\label{thm:PartitionNeighborhoodDiversity}
For a fixed graph $H$ there is an \FPT-algorithm for the \InducedHPartition problem parameterized by $\nd(G).$
\end{theorem}

\section{Preliminaries}
For a graph $G=(V,E)$ we denote by $|G|$ the number of vertices of $G,$ that is $|G| = |V|.$
For a set $U$ we denote by $\binom{U}{2}$ the set of all two element subsets of $U,$ that is $\binom{U}{2} = \{\{u,v\}\colon u,v\in U, u\neq v\}.$
Let $G= (V,E)$ be a graph and let $U\subseteq V$ the graph induced by $U$ is denoted by $G[U]$ and it is the graph $(U,E\cap\binom{U}{2}).$
Let $G = (V,E), H = (W,F)$ be graphs, we say that $G$ is {\em isomorphic} to $H,$ we denote this by $G\simeq H,$ if there exists a bijective mapping $f:V\rightarrow W$ such that $\{u,v\}\in E$ if and only if $\{f(u),f(v)\}\in F.$ 
For a graph $G=(V,E)$ and its two distinct vertices $u,v$ we say that there exists an {\em $uv$ path in $G$} if there exist edges $e_1,e_2, \dots, e_k\in E$ and vertices $v_0 = u, v_1,\dots, v_k = v\in V$ such that $e_i = \{v_{i-1},v_i\}.$
For a set of vertices $U$ we denote the set of adjacent edges as $\delta(U)$ that is the set of edges $\{\{u,v\}\colon u\in U, v\in V\setminus U\}.$
Finally a complement of a graph $G = (V, E)$ is denoted by $\bar{G}$ is a on the same vertex set $V$ with edge set $\binom{V}{2}\setminus E.$
We say that a graph $G$ is {\em connected} if there is a $uv$ path in $G$ for every two distinct vertices of $G.$
For more notation on graphs, we refer reader to a monograph by Diestel~\cite{DiestelGT}.

\subsection{Preliminaries on refuting polynomial kernels}\label{s:kernelRefutePrelim}
Here we present simplified review of a framework used to refute existence of polynomial kernel for a parameterized problem from Chapter 15 of a monograph by Cygan et al.~\cite{CFKLMPPS-FPT}.

In the following we denote by $\Sigma$ a final alphabet, by $\Sigma^*$ we denote the set of all words over $\Sigma$ and by $\Sigma^{\le n}$ we denote the set of all words over $\Sigma$ and length at most $n.$

\begin{definition}[Polynomial equivalence relation]\label{d:polyRelation}
An equivalence relation $\mathcal{R}$ on the set $\Sigma^*$ is called {\em polynomial equivalence relation} if the following conditions are satisfied:
\begin{enumerate}
  \item There exists an algorithm such that, given strings $x,y\in\Sigma^*,$ resolves whether $x\equiv_\mathcal{R} y$ in time polynomial in $|x| + |y|$.
  \item Relation $\mathcal{R}$ restricted to the set $\Sigma^{\le n}$ has at most $p(n)$ equivalence classes for some polynomial $p(\cdot).$
\end{enumerate}
\end{definition}

\begin{definition}[Cross-composition]\label{d:crossComposition}
Let $L\subseteq\Sigma^*$ be an unparameterized language and $Q\subseteq\Sigma^*\times\N$ be a parametrized language. We say that $L$ {\em cross-composes} into $Q$ if there exists a polynomial equivalence relation $\mathcal{R}$ and an algorithm $\mathcal{A},$ called the cross-composition, satisfying the following conditions. 
The algorithm $\mathcal{A}$ takes on input a sequence of strings $x_1,x_2,\dots,x_t\in\Sigma^*$ that are equivalent with respect to $\mathcal{R}$, runs in polynomial time in $\sum_{i = 1}^t |x_i|,$ and outputs one instance $(y,k)\in\Sigma^*\times\N$ such that:
\begin{enumerate}
  \item $k\le p(max_{i = 1}^t |x_i|, \log t)$ for some polynomial $p(\cdot,\cdot),$ and
  \item $(y,k)\in Q$ if and only if $x_i\in L$ for all $i.$
\end{enumerate}
\end{definition}

With this framework, it is possible to refute even stronger reduction techniques---namely polynomial compression---according to the following definition:

\begin{definition}[Polynomial compression]\label{d:polyCompression}
A {\em polynomial compression} of a parameterized language $Q\subseteq\Sigma^*\times\N$ into an unparameterized language $R\subseteq\Sigma^*$ is an algorithm that takes as an input an instance $(y,k)\in\Sigma^*\times\N,$ works in polynomial time in $|x| + k,$ and returns a string $y$ such that:
\begin{enumerate}
  \item $|y|\le p(k)$ for some polynomial $p(\cdot),$ and
  \item $y\in R$ if and only if $(x,k) \in Q.$
\end{enumerate}
\end{definition}

It is possible to refute existence of polynomial kernel using Definitions~\ref{d:polyRelation},\ref{d:crossComposition} and \ref{d:polyCompression} with the help of use of the following theorems and a complexity assumption that is unlikely to hold---namely $\NP\subseteq\coNPpoly$.

\begin{theorem}[\cite{Drucker12}]
Let $L,R\subseteq\Sigma^*$ be two languages. Assume that there exists an AND-distillation of $L$ into $R.$ Then $L\in\coNPpoly.$
\end{theorem}

\begin{theorem}
Assume that an $\NP$-hard language $L$ AND-cross-composes to a parameterized language $Q.$ Then $Q$ does not admit a polynomial compression, unless $\NP\subseteq\coNPpoly.$
\end{theorem}

\section{Preliminaries on structural graph parameters}\label{sec:StructuralGraphParameters}
We give a formal definition of several graph parameters used in this work. For a better acquaint with these parameters, we provide a map of assumed parameters in Figure~\ref{fig:parameterMap}.

One of the most restrictive graph parameters is called the {\em vertex cover number} and is defined as follows.

\begin{definition}[Vertex cover]\label{def:vc}
For a graph $G = (V,E)$ the set $U\subset V$ is called a {\em vertex cover} of $G$ if for every edge $e\in E$ it holds that $e\cap U\neq\emptyset.$

The vertex cover number of a graph, denoted as $\vc(G),$ is the least integer $k$ for which there exists a vertex cover of size $k.$
\end{definition}

We say that the vertex cover number is very restrictive graph parameter, because for a fixed positive integer $k$ the class of graphs with vertex cover number bounded by $k$ does not contain large spectra of graphs (for example some whole classes of graphs). 

\begin{figure}
  \begin{minipage}[c]{0.28\textwidth}
    \begin{tikzpicture}[node distance=1.5cm]
  \tikzstyle{arrowBasic} = [->, >=stealth, very thick]
  \tikzstyle{linDep} = [arrowBasic]
  \tikzstyle{expDep} = [arrowBasic, dashed]

  \node (cw) {$\cw$};
  
  \node[below right of=cw] (tw) {$\tw$};
  \node[below left of=cw] (mw) {$\mw$};
  
  \node[below right of=mw] (nd) {$\nd$};
  \node[below left of=mw] (tc) {$\tc$};
  \node[below of=tw] (pw) {$\pw$};

  \node[below left of=nd] (vc) {$\vc$};

  \draw[expDep] (vc) -- (nd);
  \draw[linDep] (vc) -- (tc);
  \draw[linDep] (vc) -- (pw);

  \draw[linDep] (pw) -- (tw);
  \draw[linDep] (nd) -- (mw);
  \draw[linDep] (tc) -- (mw);
  \draw[expDep] (mw) -- (cw);
  \draw[expDep] (tw) -- (cw);


  \begin{scope}[on background layer]
    \node[fill=yellow!20, fit=(vc)(pw)(tw), rounded corners] {};
  \end{scope}

  \begin{scope}[on background layer]
    \node[fill=orange!30, fit=(nd)(mw)(tc)(cw), rounded corners] {};
  \end{scope}
\end{tikzpicture}
  \end{minipage}\hfill
  \begin{minipage}[b]{0.7\textwidth}
    \caption{
A map of assumed parameters. Full arrow stands for linear upper bounds, while dashed arrow stands for exponential upper bounds.
    } \label{fig:parameterMap}
  \end{minipage}
\end{figure}
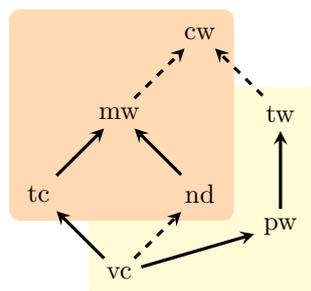

As the vertex cover number is (usually) too restrictive, many authors focused on defining other structural parameters (so as more graphs have small value of the parameter). Three most well-known parameters of this kind are the path-width, the tree-width (introduced by Robertson and Seymour~\cite{RS86:GMII}), and the clique-width (introduced by Courcelle et al.~\cite{CMR98}). Classes of graphs with bounded tree-width (respectively path-width) contain the so called sparse graph classes (i.e. the graph cannot contain too many edges).

There are (more recent) structural graph parameters which also generalize the vertex cover number but in contrary to the tree-width these parameters focus on dense graphs. First, up to our knowledge, of these parameters is the {\em neighborhood diversity} defined by Lampis~\cite{Lampis12}. We denote the neighborhood diversity of a graph $G = (V,E)$ as $\nd(G).$ 

We say that two (distinct) vertices $u,v$ are of the same {\em neighborhood type} if $N(u)\setminus\{v\} = N(v)\setminus\{u\}$ (they share their respective neighborhoods).

\begin{definition}[Neighborhood diversity~\cite{Lampis12}]\label{def:nd}
A graph $G = (V,E)$ has {\em neighborhood diversity} at most $w$ ($\nd(G)\le w$), if there exists a partition of $V$ into at most $w$ sets (we call these sets \emph{types}) such that all the vertices in each set have the same neighborhood type.
\end{definition}

Usually, we use the notion of \emph{type graph}---that is a graph $T_G$ representing the graph $G$ and its neighborhood diversity decomposition in the following way. The vertices of type graph $T_G$ (also called the \emph{template graph}) are the neighborhood types of the graph $G$ and two such vertices are joined by an edge if all the vertices of corresponding types are joined by an edge. Note that any optimal type graph is a prime graph. We would like to point out that it is possible to compute the neighborhood diversity of a graph in linear time~\cite{Lampis12}.

More recently, Ganian~\cite{Ganian11} defined the {\em twin-cover number}. We begin with an auxiliary definition. If two vertices $u,v$ have the same neighborhood type and $e = \{u,v\}$ is an edge of the graph, we say that $e$ is a {\em twin-edge}.

\begin{definition}[Twin-cover number~\cite{Ganian11}]\label{def:tc}
A set of vertices $X\subseteq V$ is a {\em twin-cover} of a graph $G=(V,E)$, if for every edge $e\in E$ either
\begin{enumerate}
  \item $X\cap e\neq\emptyset$, or 
  \item $e$ is a twin-edge.
\end{enumerate}
We say that $G$ has twin-cover number $k$ ($\tc(G) = k$) if the size of a minimum twin-cover of $G$ is $k.$
\end{definition}

Note that the twin-cover can be upper-bounded by the vertex-cover number. As the structure of graphs with bounded twin-cover is very similar to the structure of graphs with bounded vertex-cover number, there is a hope that many of known algorithms for graphs with bounded vertex-cover number can be easily turned into algorithms for graphs with bounded twin-cover.

Both previous approaches are generalized by a modular-width, defined by Gajarský et al.~\cite{GLO13}. Here we deal with graphs created by an algebraic expression that uses the following operations:
\begin{enumerate}
  \item create an isolated vertex,
  \item the disjoint union of two graphs,
that is from graphs $G = (V,E), H = (W, F)$ create a graph $(V\cup W, E\cup F)$,
  \item the complete join of two graphs,
that is from graphs $G = (V,E), H = (W, F)$ create a graph $(V\cup W, E\cup F\cup \{\{v,w\}\colon v\in V, w\in W\})$, note that the edge set of the resulting graph can be also written as $E\cup F\cup V\times W.$
  \item The substitution operation with respect to some (arbitrary) graph $T$ (for an example see Figure~\ref{fig:MWDecomposition}) with vertex set $\{v_1,v_2,\dots,v_n\}$ and graphs $G_1,G_2,\dots,G_n$ created by algebraic expression. The substitution operation, denoted by $T(G_1,G_2, \dots, G_n)$, results in the graph on vertex set $V = V_1\cup V_2\cup\dots\cup V_n$ and edge set
$$E = E_1\cup E_2\cup\dots\cup E_n\cup \bigcup_{\{v_i,v_j\}\in E(G)} V_i\times V_j,$$
where $G_i = (V_i,E_i)$ for all $i = 1,2,\dots,n.$
\end{enumerate}


\begin{figure}
  \begin{tikzpicture}[node distance=.7cm]
  \tikzstyle{small_vertex} = [circle, draw, fill=black, inner sep=1.5pt]
  \tikzstyle{under_small_graph} = [circle, inner sep=.7ex, dashed, draw]
  \tikzstyle{template_graph} = [inner sep=2ex, fill=black!10, rounded corners]
  \tikzstyle{arround_vertex} = [circle, inner sep=2pt, dashed, draw]
  \tikzstyle{old_edge} = [draw, black!40]
  \tikzstyle{tan_line} = [draw, black!70]

  \begin{scope}[local bounding box=left]
  \node[small_vertex] (L1) {};
  \node[small_vertex,above of=L1] (L2) {};
  \node[small_vertex,above of=L2] (L3) {};
  \draw (L1) -- (L2) -- (L3);
  \end{scope}
  \begin{scope}[on background layer]
    \node[under_small_graph, fit = (L1)(L2)(L3)] (lg) {};
  \end{scope}  

  \begin{scope}[local bounding box=middle, shift={($(left.east)+(1.5cm,-.4cm)$)}]
  \node[small_vertex] (M1) {};
  \node[small_vertex, above of=M1] (M2) {};
  \draw (M1) -- (M2);
  \end{scope}
  \begin{scope}[on background layer]
    \node[under_small_graph, fit = (M1)(M2)] (mg) {};
  \end{scope}  

  \begin{scope}[local bounding box=right, shift={($(middle.east)+(1.5cm,-.6cm)$)}]
  \node[small_vertex] (R1) {};
  \node[small_vertex, above of=R1] (R2) {};
  \node[small_vertex, above of=R2] (R3) {};
  \draw (R1) -- (R2);
  \end{scope}
  \begin{scope}[on background layer]
    \node[under_small_graph, fit = (R1)(R2)(R3)] (rg) {};
  \end{scope}  

  \begin{scope}[local bounding box=schema_graph, shift={($(middle.north)+(-1.5cm, 2cm)$)},node distance=1.5cm]
    \node[small_vertex] (T1) {};
    \node[small_vertex, right of=T1] (T2) {};
    \node[small_vertex, right of=T2] (T3) {};
    \draw (T1) -- (T2) -- (T3);
  \end{scope}
  \begin{scope}[on background layer]
    \node[template_graph, fit = (T1)(T2)(T3), label=right:$T$] {};
    \node[arround_vertex, fit=(T1)] (TT1) {};
    \node[arround_vertex, fit=(T2)] (TT2) {};
    \node[arround_vertex, fit=(T3)] (TT3) {};
  \end{scope}  

  \begin{scope}[local bounding box=left2, shift={($(right.west)+(5cm,-.7cm)$)}]
  \node[small_vertex] (LL1) {};
  \node[small_vertex,above of=LL1] (LL2) {};
  \node[small_vertex,above of=LL2] (LL3) {};
  \draw[old_edge] (LL1) -- (LL2) -- (LL3);
  \end{scope}
  \begin{scope}[on background layer]
    \node[under_small_graph, fit = (LL1)(LL2)(LL3)] (l2) {};
  \end{scope}  
  
  \begin{scope}[local bounding box=middle2, shift={($(left2.east)+(1.5cm,-.4cm)$)}]
  \node[small_vertex] (MM1) {};
  \node[small_vertex, above of=MM1] (MM2) {};
  \draw[old_edge] (MM1) -- (MM2);
  \end{scope}
  \begin{scope}[on background layer]
    \node[under_small_graph, fit = (MM1)(MM2)] {};
  \end{scope}

  \begin{scope}[local bounding box=right, shift={($(middle2.east)+(1.5cm,-.6cm)$)}]
  \node[small_vertex] (RR1) {};
  \node[small_vertex, above of=RR1] (RR2) {};
  \node[small_vertex, above of=RR2] (RR3) {};
  \draw[old_edge] (RR1) -- (RR2);
  \end{scope}
  \begin{scope}[on background layer]
    \node[under_small_graph, fit = (RR1)(RR2)(RR3)] {};
  \end{scope}

  \foreach \i in {(LL1),(LL2),(LL3)}
    \foreach \j in {(MM1),(MM2)}
      \draw \j -- \i;

  \foreach \i in {(RR1),(RR2),(RR3)}
    \foreach \j in {(MM1),(MM2)}
      \draw \j -- \i;

  \draw[tan_line] (tangent cs:node=TT1,point={(lg.east)},solution=2) -- (tangent cs:node=lg,point={(TT1.east)});
  \draw[tan_line] (tangent cs:node=TT1,point={(lg.west)},solution=1) -- (tangent cs:node=lg,point={(TT1.west)},solution=2);

  \draw[tan_line] (tangent cs:node=TT2,point={(mg.east)},solution=2) -- (tangent cs:node=mg,point={(TT2.east)});
  \draw[tan_line] (tangent cs:node=TT2,point={(mg.west)},solution=1) -- (tangent cs:node=mg,point={(TT2.west)},solution=2);

  \draw[tan_line] (tangent cs:node=TT3,point={(rg.east)},solution=2) -- (tangent cs:node=rg,point={(TT3.east)});
  \draw[tan_line] (tangent cs:node=TT3,point={(rg.west)},solution=1) -- (tangent cs:node=rg,point={(TT3.west)},solution=2);

  \draw[line width=2.5mm, >=stealth, ->, shorten >=.3cm,shorten <=.3cm] (rg) -- (l2);

\end{tikzpicture}
  \caption{An illustration of the modular-width decomposition of a graph. A schema of a decomposition is depicted in the left part of the picture. In the right part of the picture there is the resulting graph---gray edges represent edges from the previous step of the decomposition (with template graph $T$). The resulting graph has $\nd = 5$ and by the decomposition $\mw = 3.$}
  \label{fig:MWDecomposition}
\end{figure}
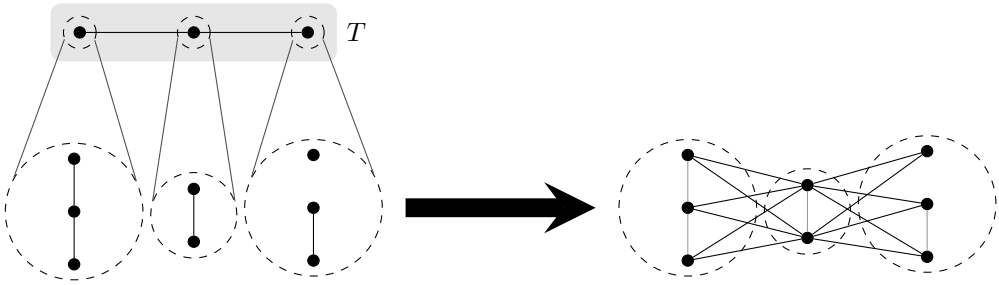

\begin{definition}[Modular-width~\cite{GLO13}]
Let $A$ be an algebraic expression that uses only operations $1$--$4$. The {\em width of expression} $A$ is the maximum number of operands used by any occurrence of operation $4$ in an expression A.

The {\em modular-width} of a graph $G$, denoted as $\mw(G)$, is the least positive integer $m$ such that $G$ can be obtained from such an algebraic expression of width at most $m.$ 
\end{definition}

When a graph $H$ is constructed by the fourth operation, that is $G = T(G_1,G_2,\dots,G_n),$ we call the graph $T$ the {\em template graph} and graphs $G_1,G_2,\dots,G_n$ are referred as {\em blocks}.
An algebraic expression of width $\mw(G)$ can be computed in linear time~\cite{TCHP08}.


\section{\InducedHPartition on graphs with bounded modular-width}
\label{sec:IHPforMWGraphs}
In this section we give a proof of the Theorem~\ref{thm:PartitionModularWidth}. We begin with a technical Lemma~\ref{lem:Hembedding} that demonstrate the limited possibilities of embedding the (fixed) graph $H$ into the graph $G$ with bounded modular-width. This exploits a close connection of the parameters neighborhood diversity and modular-width.

We then formulate the problem as a mixed integer linear problem where we can bound the number of integer variables by a function in the modular-width $\mw(G)$ of the input graph.

\subsection{Embedding of $H$ inside $G$}
The following lemma shows that the restriction for the neighborhood diversity $\nd(H) = |H|$ leads to only two possibilities of an embedding of $H$ on a particular level of a modular decomposition of the graph $G.$ 

\begin{lemma}\label{lem:Hembedding}
Let $G = T(G_1,G_2,\dots,G_n)$ be a graph and let $H$ be a graph with $|H| = \nd(H).$ For an induced subgraph $H$ in $G$ holds either 
\begin{enumerate}
  \item $H\subseteq G_i$ for some $i \in\{1,2,\dots,n\},$ or
  \item $H$ contains at most one vertex in every $G_i$ for $i = 1,2,\dots, n.$
\end{enumerate}
\end{lemma}
\begin{proof}
Assume that $H\nsubseteq G_i$ for any $i.$ Now if there are at least two vertices of $H$ in some $G_i,$ we claim that $\nd(H) < |H|$ but this finishes the proof as it contradicts assumptions of the lemma.

Let us rearrange vertices of a template graph $T$ (and corresponding graphs $G_i$) so that $H$ contains vertices of graphs $G_1,G_2,\dots,G_k$ (with $k < |H|$). Let $T'\subseteq T$ be the restriction of $T$ to vertices $v_1,v_2,\dots,v_k.$ It is easy to see that now $T'$ is a template graph for neighbourhood diversity decomposition of the graph $H,$ but (as claimed) this shows that $\nd(H)\le k < |H|.$
\end{proof}

Now with Lemma~\ref{lem:Hembedding} the most complex (an in particular, the only important) operation of a modular decomposition (operation 4) it follows that the structure inside graphs $G_1,G_2,\dots,G_n$ is not important when trying to find paths of the cover that contains vertices from more than one graph $G_i.$

So as pointed out by Gajarský et al.~\cite{GLO13} it suffices to design an algorithm for operation 4 only, as it is possible to express operations 2 and 3 using this operation. 

Note that when deciding the \InducedHPartition problem on graph $G$ with $\mw(G) < |H|,$ it follows from Lemma~\ref{lem:Hembedding} that the answer is clearly No---thus this case is trivial. So the task is to design an algorithm when $|H|\le\mw(G).$

\subsection{Mixed integer linear program}

In the following mixed integer linear program for the \InducedHPartition problem on the graph $G = T(G_1,G_2,\dots,G_n)$ the set $\mathcal{S}$ is the set of all $|H|$-tuples of vertices of the graph that form an induced copy of $H$ in $G.$ 
The set $V$ is the vertex set of the graph $G,$ moreover, as graphs $G_1,G_2, \dots, G_n$ correspond to vertices of $G,$ we will denote these as $G_v$ for $v\in V.$
The constants $w_v$ represent the number of vertices of the graph $G_v$ that are not covered by a copies of $H$ found by the previous (recursive) solution.
On the other hand, the constants $p_v$ represent the number of copies of $H$ in the recursive solution.
It may be wise to unlink some previously made $H$'s---this is why we introduce the variable $y_v,$ which expresses exactly this. An example of a situation in which this is necessary is depicted in Figure~\ref{fig:pathsMustBeDecomposed}. The program tries to cover as many vertices as possible---this is done by minimizing the number of uncovered vertices expressed by $r_v$ for a graph $G_v.$

\begin{figure}
  \begin{tikzpicture}[node distance=.7cm]
  \tikzstyle{vertex} = [circle, draw, fill=black, inner sep=1.5pt]
  \tikzstyle{block} = [draw, fill=yellow!30, inner sep=1.5ex, rounded corners]
  \tikzstyle{edge} = [draw,thick]

  \begin{scope}[xshift=2cm]
    \node[vertex] (g1v1) {};
    \foreach \i [remember=\i as \lasti (initially 1)] in {2,...,6} {
      \node[vertex, above of=g1v\lasti] (g1v\i) {};
    }
    \draw[dotted,thick] (g1v1) -- (g1v2) -- (g1v3);
    \draw[dotted,thick] (g1v4) -- (g1v5) -- (g1v6);
  \end{scope}
  \begin{scope}[on background layer]
    \node[block, fit=(g1v1)(g1v2)(g1v3)(g1v4)(g1v5)(g1v6)] (G2) {};
    \node[xshift=-.3cm, yshift=-.1cm] at (G2.north west) {$G_2$};
  \end{scope}

  \begin{scope}
    \node[vertex] (g2v1) {};
    \foreach \i [remember=\i as \lasti (initially 1)] in {2,3,4} {
      \node[vertex, above of=g2v\lasti] (g2v\i) {};
    }
    \draw[dotted,thick] (g2v2) -- (g2v3) -- (g2v4);
  \end{scope}
  \begin{scope}[on background layer]
    \node[block, fit=(g2v1)(g2v2)(g2v3)(g2v4), label=above:$G_1$] {};
  \end{scope}

  \begin{scope}[xshift=6cm,yshift=1.5cm]
    \node[vertex] (g3v1) {};
    \foreach \i [remember=\i as \lasti (initially 1)] in {2,3,4} {
      \node[vertex, above of=g3v\lasti] (g3v\i) {};
    }
  \end{scope}
  \begin{scope}[on background layer]
    \node[block, fit=(g3v1)(g3v2)(g3v3)(g3v4), label=right:$G_5$] {};
  \end{scope}

  \begin{scope}[xshift=4cm]
    \node[vertex] (g4v1) {};
    \node[vertex, above of=g4v1] (g4v2) {};
    \node[vertex, right of=g4v1] (g4v3) {};
    \node[vertex, right of=g4v3] (g4v4) {};
    \node[vertex, right of=g4v4] (g4v5) {};
    \draw[dotted,thick] (g4v3) -- (g4v4) -- (g4v5);
  \end{scope}
  \begin{scope}[on background layer]
    \node[block, fit=(g4v1)(g4v2)(g4v3)(g4v4)(g4v5), label=right:$G_3$] {};
  \end{scope}

  \begin{scope}[xshift=4cm, yshift=3cm]
    \node[vertex] (g5v1) {};
    \node[vertex, above of=g5v1] (g5v2) {};
  \end{scope}
  \begin{scope}[on background layer]
    \node[block, fit=(g5v1)(g5v2), label=below:$G_4$] {};
  \end{scope}

  \draw[edge] (g2v1) -- (g1v1) -- (g4v2);
  \draw[edge] (g2v2) -- (g1v3) -- (g4v1);
  \draw[edge] (g2v3) -- (g1v4) -- (g3v1);
  \draw[edge] (g2v4) -- (g1v2) -- (g3v2);
  \draw[edge] (g3v3) -- (g5v1) -- (g1v5);
  \draw[edge] (g3v4) -- (g5v2) -- (g1v6);

  \begin{scope}[xshift=9.5cm,yshift=2cm]
    \node[vertex,label=above:$v_1$] (v1) {}; 
    \node[vertex,right=of v1,label=above:$v_2$] (v2) {}; 
    \node[vertex,right=of v2,label=below:$v_3$] (v3) {};
    \node[vertex,above=of v3,label=left:$v_4$] (v4) {};
    \node[vertex,right=of v4,label=below:$v_5$] (v5) {};
  
    \draw[ultra thick] (v1) -- (v2);
    \draw[ultra thick] (v5) -- (v2);
    \draw[ultra thick] (v3) -- (v2);
    \draw[ultra thick] (v2) -- (v4) -- (v5);

    \node[below=of v2] {The template graph};
  \end{scope}
\end{tikzpicture}
  \caption{An example of a graph created by a modular decomposition. Template graph is given on the right. Paths inside blocks (precomputed by an induction) are showed as dotted, while an optimal solution for this graph is drawn with solid edges. The graph shows that it is essential to allow unlinking of previously constructed paths.}
  \label{fig:pathsMustBeDecomposed}
\end{figure}
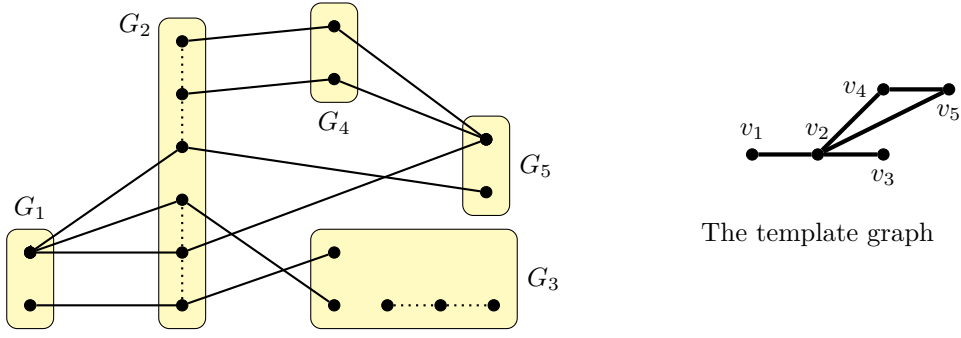

\textbf{Mixed integer linear program}
\begin{equation*}
\begin{aligned}
& \text{minimize}
& & \sum_{v \in V} r_i & \\
& \text{subject to}
& & r_v = w_v + t\cdot y_v - \sum_{S\in\mathcal{S}: S\ni v} x_S & \forall v\in V \\
& & & y_v \le p_v & \forall v\in V \\
& \text{where}
& & x_S\in\N & \forall S\in\mathcal{S} \\
& & & y_v\in\N & \forall v\in V \\
& & & r_v\ge 0 & \forall v\in V \\
\end{aligned}
\end{equation*}

The total number of integral variables may be upper-bounded by $|T| + |T|^{|H|},$ so if we denote by $n$ the size of the template graph (and thus the modular-width of an input graph) the upper bound can be expressed as $n + n^n.$ Now we can apply the result of Lenstra~\cite{lenstra83} (with an enhancement due to Frank and Tardos~\cite{FT87}):

\begin{proposition}
Let $p$ be the number of integral variables in a Mixed integer linear program and let $L$ be the number of bits needed to encode the program. Then it is possible to find an optimal solution in time $\BigO(p^{2.5p}\poly(L))$ and a space polynomial in $L.$
\end{proposition}

\subsection{Refuting polynomial kernels}
In this section we prove Theorem~\ref{thm:IHPKernelRefuteMW}. First observe that it suffices to prove the theorem only for connected graph $H.$ This trivial as (as for clique-width) for every graph $G$ it holds that $\mw(G) = \mw(\bar{G})$---and thus we can ask the question in complement setting.

As a polynomial equivalence relation we take the following relation. Two instances $(G_1, H_1), (G_2,H_2)$ are equivalent if $|G_1| = |G_2|$ (that is they have the same number of vertices) and if $H_1 \simeq H_2.$ It is easy to see that this defines a polynomial equivalence relation (together with the class of malformed instances---instances that do not encode graph).

Now observe that if we take a disjoint union of two graphs $G,H$ then $\mw(G\mathbin{\dot{\cup}} H) \le \max\{|G|, |H|\}.$ This is not hard to see as we can take $G$ to be a type graph for $G$ and similarly $H$ to be a type graph for $H.$ Then the disjoint union does not change the modular-width as the second operation of modular decomposition is exactly the disjoint union and does not change the width of the decomposition. So by an inductive argument for graphs $G_1,G_2,\dots,G_n$ it holds that $\mw(G_1 \mathbin{\dot{\cup}} G_2 \mathbin{\dot{\cup}} \cdots \mathbin{\dot{\cup}} G_t) \le \max_{1\le i\le t}\{|G_i|\}.$

As the designed polynomial equivalence relation assures that for all $i$ and $j$ the graphs $H_i\simeq H_j$ we write $H$ to be the common graph for all instances.
A cross-composition of equivalent instances $(G_1,H),(G_2, H),\dots, (G_t,H)$ we take as an instance $G$ simply the disjoin union of all these graphs. Now this new instance $(G,H)$ has partition into copies of $H$ if and only if all instances $(G_i, H)$ have this partition. And the previous paragraph shows that $\mw(G)\le\max_{1\le i \le t}\{|G_i|\}.$
And thus this finishes the design of AND-distillation and the proof of Theorem~\ref{thm:IHPKernelRefuteMW} for all graphs $H$ for which the \InducedHPartition problem is \NP-hard.

\subsection{\InducedHPartition on graphs with bounded neighborhood diversity}

In this section we will present a proof of Theorem~\ref{thm:PartitionNeighborhoodDiversity}. We begin by with a lemma that allows us to focus on simpler and highly structured instances. The rest of the proof is by bounding the number of possibilities of embedding a graph $H$ into $G$ and finally using a integer linear programming for finding the solution for the \InducedHPartition problem.

Note that here it is possible that the embedding of a graph $H$ into a graph $G$ does not have to obey the neighborhood diversity decomposition---it is possible that for example a clique type of $H$ may be embedded among several clique (or even independent) types of $G.$ For an example of such a situation see Figure~\ref{fig:embeddingNDintoND}.

The proof is very similar to the proof of Lemma~\ref{lem:Hembedding}.

\begin{lemma}\label{lem:EmbeddingNDintoND}
Let $G, H$ be graphs such that it is possible to partition $G$ into copies of $H$. Then $\nd(H)\le\nd(G).$
\end{lemma}

\begin{proof}[Proof of Lemma~\ref{lem:EmbeddingNDintoND}]
It suffices to prove that the partition embedding has to obey the neighborhood diversity decomposition of the graph $H.$ For the rest of the proof we fix one copy of $H$ and $G.$

To see this take the type graphs $T_G, T_H$ and assume an embedding of $H$ into $G.$ This is no more a function as now it is possible that for a single vertex of $T_H$ there are more vertices of $T_G$---we may choose arbitrarily among these vertices, so that we obtain an embedding of $T_H$ into $T_G$ as a function---we denote this embedding as $\varphi: V(T_H)\rightarrow V(T_G).$

Now assume that the function $\varphi$ is not injective. This is absurd as this would prove that the graph $T_H$ is not a prime graph.
\end{proof}

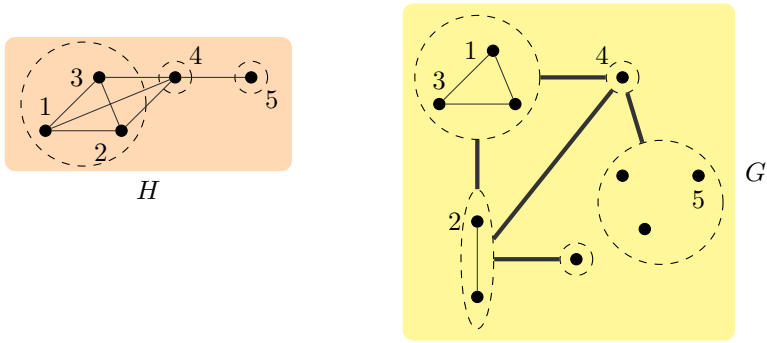
\begin{figure}
  \begin{tikzpicture}[node distance=1cm]
  \tikzstyle{vertex} = [circle, draw, fill=black, inner sep=1.5pt]
  \tikzstyle{underH} = [rounded corners, fill=orange!30, inner sep=3ex]
  \tikzstyle{underG} = [rounded corners, fill=yellow!50, inner sep=1ex]
  \tikzstyle{Ctypes} = [circle, dashed, draw, inner sep=2pt]
  \tikzstyle{Etypes} = [ellipse, dashed, draw, inner sep=2pt]
  \tikzstyle{edge} = [draw, black!80]
  \tikzstyle{macro_edge} = [ultra thick, draw, black!80]

  \begin{scope}[local bounding box=graphH]
    \node[vertex, label=above:$1$] (hv1) {};
    \node[vertex, right of=hv1, label=below left:$2$] (hv2) {};
    \node[vertex, above right of=hv1, label=left:$3$] (hv3) {};
    \node[vertex, right of=hv3, label=above right:$4$] (hv4) {};
    \node[vertex, right of=hv4, label=below right:$5$] (hv5) {};

    \draw[edge] (hv1) -- (hv2) -- (hv3) -- (hv1) -- (hv4) -- (hv5);
    \draw[edge] (hv2) -- (hv4) -- (hv3);
  \end{scope}
  \begin{scope}[on background layer]
    \node[fit=(hv1)(hv2)(hv3)(hv4)(hv5), label=below:$H$, underH] {};
    \node[Ctypes, fit=(hv1)(hv3)(hv2)] (ht1) {};
    \node[Ctypes, fit=(hv4)] (ht3) {};
    \node[Ctypes, fit=(hv5)] (ht4) {};
  \end{scope}

  \begin{scope}[local bounding box=graphG, shift={($(graphH.east) + (2cm, 0cm)$)}]
    \node[vertex, label=above:$3$] (gt1v1) {};
    \node[vertex, right of=gt1v1] (gt1v2) {};
    \node[vertex, above right of=gt1v1, label=left:$1$] (gt1v3) {};
    \draw[edge] (gt1v1) -- (gt1v2) -- (gt1v3) -- (gt1v1);
    \node[Ctypes, fit=(gt1v1)(gt1v2)(gt1v3)] (gt1) {};

    \node[vertex, below=of gt1, label=left:$2$] (gt2v1) {};
    \node[vertex, below of=gt2v1] (gt2v2) {};
    \draw[edge] (gt2v1) -- (gt2v2);
    \node[Etypes, fit=(gt2v1)(gt2v2)] (gt2) {};

    \node[vertex, right=of gt1, label=above left:$4$] (gt3v1) {};
    \node[Ctypes, fit=(gt3v1)] (gt3) {};

    \node[vertex, below=of gt3] (gt4v1) {};
    \node[vertex, right of=gt4v1, label=below:$5$] (gt4v2) {};
    \node[vertex, below left of=gt4v2] (gt4v3) {};
    \node[Ctypes, fit=(gt4v1)(gt4v2)(gt4v3)] (gt4) {};

    \node[vertex, right=of gt2] (gt5v1) {};
    \node[Ctypes, fit=(gt5v1)] (gt5) {};

    \draw[macro_edge] (gt1) -- (gt2) -- (gt3) -- (gt1);
    \draw[macro_edge] (gt4) -- (gt3);
    \draw[macro_edge] (gt2) -- (gt5);
  \end{scope}
  \begin{scope}[on background layer]
    \node[underG, fit=(gt1)(gt2)(gt3)(gt4)(gt5), label=right:$G$] {};
  \end{scope}
\end{tikzpicture}
  \caption{An example of embedding of a graph $H$ into a graph $G$---vertices are labeled $1,2,\dots,5$ in both to show an embedding. The types in graphs are indicated by dashed circles around a group of vertices. Bold edges between types represent a complete bipartite graph.}\label{fig:embeddingNDintoND}
\end{figure}

Now it is easy to see that the algorithm for the \InducedHPartition problem we may assume that $\nd(H)\le\nd(G)$ and that $H$ is a connected graph.
This allows us to simply build an integer linear program for this problem.

As the graph $H$ is fixed in our setting the number of vertices $n_H$ of $H$ is not a part of the input. And so one particular embedding of $H$ into $G$ can be described by specifying a type to which a particular vertex of $H$ is mapped. So there are at most $\nd(G)^n_H$ possibilities of embedding $H$ into $G.$ Thus the straightforward integer linear program has bounded dimension---this finishes the proof of Theorem~\ref{thm:PartitionNeighborhoodDiversity}.

\section{Connected graph partition problem is \texorpdfstring{MSO$_2$}{MSO2} definable}
\label{sec:connGIPisMSO2}
In this section we show that for a fixed connected graph $H$ can be described by a MSO$_2$ formula. Even though this is not obvious at the first sight. The straightforward expression by a formula seems to operate with copies of $H$ inside $G$---but this number cannot be bounded in terms of $\tw(G)$ and $|H|.$

We will proceed as follows. We will describe the solution to the $\InducedHPartition$ problem as a property of a set of edges which are in the solution. This approach is not expressible in MSO$_1.$
We describe the property of being solution as ``every connected part of the solution is isomorphic to $H$''. This we express by discovering the copy of $H$ from a particular (fixed in advance) vertex in $H.$
So we have to start with identifying this vertex.

In the following the set of edges $F$ (of a graph $G$) is the desired solution of the \InducedHPartition problem.

First of all we would like to give an evidence that it is not possible to express the \InducedHPartition problem by an MSO$_1$ formula. Roughly speaking about the expressive power of MSO$_1$ logic it is impossible to distinguish between two large cliques~\cite{LibkinFMT}---namely for every MSO$_1$ formula $\varphi$ there is a positive integer $N$ such that it is impossible to distinguish two cliques $K_{N}$ and $K_{N+1}.$ If we build two graphs $G_1, G_2$ as $G_1 = K_{N-1}\cup K_{N+1}$ and $G_2 = K_{N}\cup K_{N}$ for large enough $N$ that is divisible by $3$ the framework of Ehrenfeucht--Fraïssé games graphs $G_1$ and $G_2$ give that the \TrianglePartition problem is not expressible in MSO$_1$ logic. It is possible to generalize this ideal to other graphs as well.

\subsection{Identifying a vertex in $H$}
For two distinct vertices $u,v$ in a graph $H,$ the distance $d_H(u,v)$ denotes the least length of an $u-v$-path in the graph $H.$
By a {\em diameter} of a graph $H = (V,E)$ we denote the number $\diam(H) = \max_{u,v\in V} d_H(u,v).$

Let $F$ be the assumed solution.
For a fixed graph $H$ given a vertex $u$ of a graph $G$ it is possible to find a particular vertex $v$ in some connected component of $F$ by an MSO$_2$ formula 
$$
\varphi_{u,v} := (\forall u\in V,\,\exists v\in V)(\exists e_1,e_2,\dots,e_d\in F)(u\in e_1\wedge e_i \cap e_{i+1}\neq\emptyset\wedge v\in e_d),
$$
where $d = \diam(H).$ Note that our (sub)formula allows us do repeat edges and that we need at most $\diam(H)$ edges to find a desired vertex from any vertex inside the graph $H.$

\subsection{Recognition of a copy of $H$ from a particular vertex $v$}
We now may assume that $v$ is a vertex of $H$ that was chosen in advance and fixed. We will proceed by identifying all vertices of a copy of $H$ that contains $v.$ Then we may inscribe that all edges and non-edges describing a copy of $H$ are present. 

$$
\varphi_{v,H} := (\exists u_1,u_2,\dots u_{|H|})(u_1 = v\wedge
\bigwedge_{\{u_i,u_j\}\in E(H)} \{u_i,u_j\}\in F
\bigwedge_{\{u_i,u_j\}\notin E(H)} \{u_i,u_j\}\notin F)
$$

Here the trick is that we may assume a particular enumeration of vertices of $H$ chosen in advance. We identify the set of vertices of a copy of $H$ as $(\exists U := \bigcup_{i = 1}^{|H|} u_i)$ in the following.
The last thing is to show that no other vertex is connected to those forming this particular copy of $H$ inside the solution $F.$

$$
\varphi_{F,H} := (\forall e\in F\colon (\exists x,y\in U\colon x\in e\,\&\,y\in e)\vee\forall x\in U x\notin e)
$$

The final formula may be written as
$$
\varphi_H := (\exists F)(\varphi_{u,v}\wedge\varphi_{v,H}\wedge\varphi_{F,H}).
$$

This finishes the proof, as $\varphi_H$ is an MSO$_2$ formula and so is testable in \FPT time (for a fixed connected graph $H$) on graph with bounded treewidth by Courcelle's theorem~\cite{Courcelle90}. This finishes the proof of Theorem~\ref{thm:PartitionMSOFormula}.

\section{Conclusions}
In this section we enclose our paper. We have studied the \InducedHPartition problem from the parametrized point of view, but there are some open problems to which we would like to give a brief description. We divide the section according to parameters having the leading role in the posted open problems.

A first and the most natural question is to give a definition of other well-known problems as the \InducedHPartition problem for a fixed graph $H$ (or for tuples of $H$'s or even for some classes of graphs).

\subsection{General graphs}
Another very interesting question is what is the complexity of a similar problem to the \InducedHPartition problem---we can drop the induced condition. We call it the \HPartition problem. We would like to ask the question about the (parameterized) complexity of this problem.

For graphs $G = (V,E),H=(W,F)$ with $|V| = |W|\cdot n,$ we say that it is possible to {\em partition $G$ into non-induced copies of $H$} if there exist sets $V_1,V_2,\dots,V_n$ such that
\begin{itemize}
  \item $|V_i| = |W|$ for every $i = 1,2,\dots, n,$ 
  \item $\bigcup_{i = 1}^n V_i = V,$ and
  \item $H \simeq G_i' \subseteq G[V_i]$ for every $i = 1,2,\dots, n.$
\end{itemize}


\begin{prob}{\HPartition}
FIXED: & Template graph $H.$ \\
INPUT: & Graph $G$  with $|G| = n\cdot|H|$ for an integer $n.$ \\
QUESTION: & Is there a partition of $G$ into non-induced copies of $H?$ \\
\end{prob}

Note for example that the \HPartition problem is trivial for graphs without edges---that is in the case $H\simeq k\cdot K_1$ for some positive integer $k$---i.e. the answer is always ``Yes''. On the contrary in the case of $H\simeq K_k$ both problems coincides and thus in this case the \HPartition problem is \NP-complete. What is the overall picture---that is it possible to identify a property (all properties) a graph $H$ has to fulfill so that the \HPartition problem admits a polynomial-time algorithm?

\subsection{Sparse graphs}
We have shown in Section~\ref{sec:connGIPisMSO2} that it is possible to express the solution to the \InducedHPartition problem (for a connected graph $H$) by an MSO$_2$ formula. We conjecture that
\begin{conjecture}
For every disconnected graph $H$ the \InducedHPartition problem is $\mathsf{W}[1]$-hard when parameterized by the tree-width of the input graph $G.$
\end{conjecture}

As we have mentioned the \EquitableColoring problem fits well in our setting of partitioning of a graph. This problem was shown to be \W[1]-hard by Fellows et al.~\cite{Fellows11:Colorful} when parameterized by the tree-width of the input graph. A strengthening of this fact---the \EquitableColoring problem is \W[1]-hard with respect to tree-depth
was proven by Gajarský et al.~\cite{GLO13}. We would like to ask a question concerning these graph models, namely the tree-depth. Is there a disconnected graph $H$ such that the \InducedHPartition problem is \FPT with respect to the tree-depth of these parametrizations (different from $2\cdot K_1$)?

\subsection{Dense graphs}
We have presented in Section~\ref{sec:IHPforMWGraphs} an algorithm for a fixed graph $H$ from a certain class of graphs. Is it possible to extend this result to a broader class of graphs $H$? Most important form this point of view seem graphs on $3$ vertices---a path $P_3$ and a triangle $K_3$ (the rest of $3$ vertex graphs would be resolved using complements).

For all these graphs $H = P_3$ or $H = K_3$ the neighbourhood diversity is lesser than the number of vertices---so the structure of smaller solutions seems to be very important (is it possible to keep all/important solutions?).

Another important task in this area is to understand the boundary (viewed from the parameterized complexity point of view) between modular-width and neighborhood diversity, twin-cover and clique-width. We hope that our knowledge in this area can be extended in the highlight of the \InducedHPartition problem---namely we should identify graphs $H$ with the property that on one parameter the problem is fixed parameter tractable while on the other hand it is \W[1]-hard on some other parameter (higher in the parameter hierarchy).

Finally, our techniques from Lemma~\ref{lem:EmbeddingNDintoND} showed that the \InducedHPartition problem admits an \FPT algorithm when the graph $H$ is a prime graph even when the graph $H$ is a part of the (restricted) input. More generally, there is an \FPT algorithm if we extend the graph class by allowing constant number of vertices inside every type of the graph $H.$
A is it possible to show an \FPT algorithm parametrized by neighborhood diversity of graph $G$ that takes the graph $H$ as input?

Unlike neighborhood diversity, for which we have proven the \InducedHPartition problem has an \FPT algorithm, nothing besides $H$ being a prime graph (which is implied by the Theorem~\ref{thm:PartitionModularWidth}) is known with parametrization by the twin-cover. It is an interesting question (similar as for modular-width) whether there is a graph $H$ that distinguishes neighborhood diversity from twin-cover from the perspective of the \InducedHPartition problem.

Our results give possibilities for parametrized algorithms with respect to clique-width---namely for the class of prime graphs. Here we would like to propose a concrete question for the \InducedHPartitionH{$P_5$}---is there an \FPT algorithm for this problem with respect to clique-width?

\bibliography{main}


\end{document}